\documentclass[journal,10pt]{IEEEtran}

\usepackage{pgfplots}

\pgfplotsset{compat=1.10}

\usepgfplotslibrary{fillbetween}

\usepackage{bm}
\usepackage{lipsum}
\usepackage{makeidx}
\usepackage{enumerate}
\usepackage{color}
\usepackage{cite}
\usepackage{amsmath,amsthm}   
\usepackage{amssymb}
\usepackage{nomencl}
\usepackage{multirow}
\usepackage{graphicx}
\usepackage{multirow}
\usepackage{anysize}
\usepackage{float}
\usepackage{epstopdf}
\usepackage{threeparttable}
\usepackage{multicol}
\DeclareGraphicsExtensions{.pdf,.jpeg,.png,.jpg,.emf,.eps}
\usepackage{calc}
\usepackage{here}
\usepackage{url}

\usepackage{upref}
\usepackage{comment}
\usepackage{times}
\usepackage{dsfont}
\usepackage{epic,eepic}
\usepackage{rawfonts}
\usepackage[T1]{fontenc}
\usepackage{latexsym}
\usepackage{amsfonts}

\usepackage{capitalgreekitalic}

\usepackage{geometry}
\usepackage[font=small,labelfont=bf]{caption}
\usepackage[font=small,labelfont=bf]{subcaption}

\usepackage{tikz,pgfplots}
\usetikzlibrary{calc}
\usetikzlibrary{intersections}
\usetikzlibrary{arrows,shapes}
\usetikzlibrary{positioning}

\newtheorem{lemma}{Lemma}

\newtheorem{remark}{Remark}

\theoremstyle{definition}

\newlength\figureheight
\newlength\figurewidth
\begin{document}

\title{Improper Gaussian Signaling for the $K$-user MIMO Interference Channels with Hardware Impairments}

\author{Mohammad Soleymani$^*$, \emph{Student Member, IEEE}, Ignacio Santamaria$^\dag$ \emph{Senior Member, IEEE} and\\ Peter J. Schreier$^*$, \emph{Senior Member, IEEE}
\\ \thanks{
$^*$Mohammad Soleymani and Peter J. Schreier are with the Signal and System Theory Group, Universit\"at Paderborn,  33098 Paderborn, Germany, http://sst.upb.de  (emails: \protect\url{{mohammad.soleymani,  peter.schreier}@sst.upb.de}).  

$^\dag$Ignacio Santamaria is with the Department of Communications Engineering, University of Cantabria, 39005 Santander, Spain (email: \protect\url{i.santamaria@unican.es}).
}}
\maketitle
\begin{abstract}
This paper investigates the performance of improper Gaussian signaling (IGS) for the $K$-user multiple-input, multiple-output (MIMO)  interference channel (IC) with hardware impairments (HWI). HWI may arise due to imperfections in the devices like I/Q imbalance, phase noise, etc. With I/Q imbalance, the received signal is a widely linear transformation of the transmitted signal and noise. Thus, the effective noise at the receivers becomes improper, which means that its real and imaginary parts are correlated and/or have unequal powers. 

IGS can improve system performance with improper noise and/or improper interference. In this paper, we study the benefits of IGS  for this scenario in terms of two performance metrics: achievable rate and energy efficiency (EE). We consider the rate region, the sum-rate, the EE region and the global EE optimization problems to fully evaluate the IGS performance. To solve these non-convex problems, we employ an optimization framework  based on majorization-minimization algorithms, which allow us to obtain a stationary point of any optimization problem  in which either the objective function and/or constraints are linear functions of rates. Our numerical results show that IGS can significantly  improve the performance of the  $K$-user MIMO  IC with HWI and I/Q imbalance, where its benefits increase with the number of users, $K$, and the imbalance level, and decrease with the number of antennas.

\end{abstract} 
\section{Introduction}
Wireless communication devices are never completely ideal in practice, which can significantly degrade the system performance especially when the hardware non-idealities are not adequately modeled and accounted for the system design. 
In general, hardware impairments (HWI) may occur due to imperfections such as quantization noise, phase noise, amplifier nonlinearities, and I/Q imbalance \cite{ bjornson2013new, tlebaldiyeva2019device,    javed2017full, javed2018improper, 
soleymani2019improper, boshkovska2018power, zhang2019secure,  zhu2017analysis, bjornson2014massive, zhang2016achievable, zhang2018scaling, xia2015hardware,  feng2019two, cheng2018performance,
bjornson2013capacity,      zhang2018performance, younas2017bandwidth,  
mohamed2019massive,  tlebaldiyeva2019performance}. 
In addition to hardware non-idealities, communication systems may suffer from strong interference because modern wireless communication systems are mostly interference-limited. Thus, interference-management techniques will play a key role in  
5G  and future generations of wireless communication systems \cite{andrews2014will}.  
In the last decade, the use of improper Gaussian signaling (IGS) has been proposed and extensively studied as an interference-management technique \cite{ javed2017full, cadambe2010interference, yang2014interference, lameiro2013degrees, ho2012improper, zeng2013transmit, soleymani2019robust, soleymani2019ergodic, nguyen2015improper, lagen2016superiority, kurniawan2015improper, lameiro2017rate, lameiro2015benefits, amin2017overlay, soleymani2019energy,  Sole1909:Energy, lagen2016coexisting, gaafar2017underlay, nasir2019improper, tuan2019non}. 
The real and imaginary parts of complex improper signals are correlated and/or have unequal powers, for a full treatment of improper signals the reader is referred to \cite{schreier2010statistical,adali2011complex,adali2014optimization}. 
 
\subsection{Related work}
The impact of hardware imperfections has been studied for various wireless communications scenarios in \cite{boshkovska2018power, soleymani2019improper, bjornson2013new, javed2017full, javed2018improper, tlebaldiyeva2019device,
zhu2017analysis, zhang2016achievable, xia2015hardware, bjornson2014massive, cheng2018performance,
bjornson2013capacity,   zhang2018scaling, feng2019two, zhang2019secure, zhang2018performance, younas2017bandwidth, mohamed2019massive, tlebaldiyeva2019performance}. 
For instance, \cite{soleymani2019improper,bjornson2013new, javed2017full, javed2018improper, tlebaldiyeva2019device} considered different interference-limited networks with single-antenna transceivers subject to additive hardware distortions (AHWD). 
When there is AHWD, the noise distortion power is a linear function of the signal power at the corresponding antenna \cite{soleymani2019improper,bjornson2013new, javed2017full, javed2018improper, tlebaldiyeva2019device}.
In \cite{bjornson2013new}, the authors investigated the effect of AHWD on the performance of a dual-hop relay with both amplify-and-forward and decode-and-forward protocols, and derived 
closed-form expressions for the outage probabilities, as well as an upper bound for the ergodic capacity. 
The outage probability for a device-to-device millimeter wave communication system with complex proper AHWD was derived in \cite{tlebaldiyeva2019device}. 
However, AHWD may be, in general, improper due to I/Q imbalance\footnote{In this paper, AHWD noise refers to the model in \cite{boshkovska2018power,soleymani2019improper,
zhu2017analysis, zhang2016achievable, xia2015hardware,bjornson2014massive,
bjornson2013capacity,bjornson2013new} in which the power of the AHWD noise is a linear function of the power of the received signal. On the other hand, HWI with I/Q imbalance refers to the model in \cite{javed2019multiple,boulogeorgos2016energy}, where the received signal is a function of the widely linear transform of the transmitted signal and noise, and the variance of the noise  is fixed and independent of signal powers.} 
\cite{javed2019asymmetric, javed2019multiple, javed2018improper, javed2017asymmetric, javed2017impact, javed2017full, soleymani2019improper,boulogeorgos2016energy}.
The authors in \cite{javed2017full} considered a relay channel with improper AHWD and maximized the achievable rate of the system by optimizing over complementary variances. 
In \cite{javed2018improper}, the authors considered a full-duplex multihop relay channel with improper AHWD. 
The work in \cite{soleymani2019improper} considered the 2-user single-input, single-output (SISO)  IC with improper AHWD and proposed two suboptimal IGS schemes to obtain the achievable rate region of the  2-user SISO IC.

Hardware non-idealities can be even more critical in multiple-antenna systems \cite{boshkovska2018power, zhang2016achievable, xia2015hardware, bjornson2014massive, cheng2018performance, zhu2017analysis,
bjornson2013capacity,   zhang2018scaling, feng2019two, zhang2019secure, zhang2018performance, younas2017bandwidth}. 
In \cite{bjornson2013capacity}, the authors studied the capacity limit and multiplexing gain of  multiple-input, multiple-output (MIMO) point-to-point systems with AHWD at both transmitter and receiver sides.
The papers  \cite{boshkovska2018power, zhang2019secure, zhu2017analysis} studied secure communications for massive MIMO systems with AHWD in different scenarios. 
The paper  \cite{bjornson2014massive} investigated the impact of AHWD on the performance of  cellular communication  systems in which the base station employs a massive number of antennas.
The papers  \cite{zhang2018scaling,zhang2016achievable,xia2015hardware} studied the performance of massive MIMO systems with AHWD in fading channels in different scenarios. 
In \cite{feng2019two}, the authors investigated the system performance of a two-way massive MIMO relay channel with AHWD. 
In \cite{cheng2018performance}, the authors considered beamforming designs for a dual-hop massive MIMO amplify-and-forward relay channel in the presence of AHWD and analyzed the outage probabilities for the system. 

In addition to AHWD, there might be other sources of hardware imperfections like I/Q imbalance. When I/Q imbalance occurs, the received signal can be modeled through a widely linear transformation  of the transmitted signal and the aggregated noise. 
Hence, the received signal can be improper even if the transmitted signal and additive noise are both proper \cite{javed2019multiple}.
It has been shown that IGS can improve the system performance in the presence of improper noise or interference-plus-noise \cite{javed2019asymmetric, javed2019multiple, javed2018improper, javed2017asymmetric, javed2017impact, javed2017full, soleymani2019improper,boulogeorgos2016energy}. 
For example, it is shown in \cite{javed2017asymmetric} that IGS is the optimal signaling for a point-to-point single-input, multiple-output (SIMO) system with asymmetric or improper AHWD.

Improper signaling can  be also used  as an 
interference-management technique in modern wireless communications systems \cite{ javed2017full, cadambe2010interference, yang2014interference, lameiro2013degrees, ho2012improper, zeng2013transmit, nguyen2015improper, lagen2016superiority, kurniawan2015improper, lameiro2017rate, lameiro2015benefits, amin2017overlay, soleymani2019energy, soleymani2019robust, soleymani2019ergodic, Sole1909:Energy, lagen2016coexisting, gaafar2017underlay, nasir2019improper, tuan2019non}. 
 It was  shown that IGS can improve several performance metrics of different interference-limited systems. 
First, interference-alignment techniques transmitting IGS can increase the degrees of freedom (DoF) of different ICs, as proved in \cite{ cadambe2010interference, yang2014interference,lameiro2013degrees}. 
Second, IGS can provide significant gains in terms of achievable rate and/or power/energy-efficiency perspectives when treating interference as noise (TIN) is the decoding strategy \cite{ho2012improper, zeng2013transmit, nguyen2015improper, lagen2016superiority, kurniawan2015improper, lameiro2017rate, lameiro2015benefits, amin2017overlay, soleymani2019energy, soleymani2019robust, Sole1909:Energy, lagen2016coexisting, gaafar2017underlay, nasir2019improper, tuan2019non, soleymani2019ergodic}. 
IGS was considered as an interference-management technique for the first time in \cite{cadambe2010interference}, where it was shown that IGS can increase the DoF of the 3-user IC. 
The papers  \cite{ho2012improper, zeng2013transmit, soleymani2019robust, soleymani2019ergodic} showed the superiority of IGS in the 2-user SISO IC in terms of achievable rate. 
A robust IGS design for the 2-user IC with imperfect channel state information (CSI) was proposed in \cite{soleymani2019robust }. 
The ergodic rate of IGS and proper Gaussian signaling (PGS) schemes in the 2-user IC with statistical CSI was studied in \cite{soleymani2019ergodic}. 
The works  \cite{nguyen2015improper, lagen2016superiority, kurniawan2015improper, lameiro2017rate}  investigated performance improvements by IGS in the Z-IC. 
 In \cite{amin2017overlay}, the authors showed that IGS can decrease the outage probability of the secondary user (SU) in an overlay cognitive radio (CR) for a given rate target. 
The work  \cite{lameiro2015benefits} showed that IGS can increase the achievable rate of the SU in an underlay CR (UCR) if the power gain of the interference link is greater than a threshold.  
Energy-efficient designs for IGS were proposed in \cite{soleymani2019energy} for UCR and in \cite{Sole1909:Energy} for the  $K$-user SISO IC. 
In \cite{nasir2019improper}, the authors showed that IGS can increase the minimum achievable rate of the users in the MIMO broadcast channel. 
The paper  \cite{ tuan2019non} investigated the performance improvements by IGS in non-orthogonal multiple access systems.

The papers  \cite{cadambe2010interference, yang2014interference, ho2012improper, zeng2013transmit, nguyen2015improper, lagen2016superiority, kurniawan2015improper, lameiro2017rate, lameiro2015benefits, amin2017overlay, soleymani2019energy, soleymani2019robust, soleymani2019ergodic, Sole1909:Energy, lagen2016coexisting, gaafar2017underlay, nasir2019improper, tuan2019non} studied the performance of IGS with ideal devices, which is not a realistic scenario. 
The papers  \cite{javed2019asymmetric, javed2018improper, javed2017asymmetric, javed2017impact, javed2017full, soleymani2019improper} consider the performance of IGS with AHWD as indicated before. However, to the best of our knowledge, there is no work on IGS in multiple-antenna interference-limited systems in the presence of  HWI with I/Q imbalance.

The performance of IGS in the 2-user and/or $K$-user SISO ICs has been vastly studied, e.g., in \cite{ cadambe2010interference, ho2012improper, zeng2013transmit, soleymani2019robust, soleymani2019ergodic, nguyen2015improper,    Sole1909:Energy}.  
However, 
the performance of IGS in the $K$-user MIMO IC  still requires further investigation.  
In SISO systems, it is known that the benefits of IGS are greatly reduced when the number of resources, e.g., time or frequency channel uses, increases. For instance, in \cite{soleymani2018improper}, we showed that IGS does not provide a significant gain in orthogonal frequency division multiplexing (OFDM) UCR systems when the number of subcarriers grows. 
The same behavior is observed in \cite{hellings2018improper}, where PGS is proved to be optimal in the 2-user IC if coded time-sharing is allowed in which the average power consumption is constrained instead of the {\em instantaneous power}. 
Hence, it seems that increasing the number of temporal or frequency dimensions provides a more flexible power allocation for PGS, which might lead to minor improvements by IGS. 
In MIMO systems,  the number of resources increases by allowing more antennas at the transceivers. Thus, the following questions arise: how does IGS perform in the $K$-user MIMO ICs? Is IGS still beneficial when the number of spatial dimensions (antennas) increases?  In this paper, we answer these questions and analyze the performance of IGS by considering different rate and energy-efficiency metrics and solving various optimization problems.

\subsection{Contribution}
This paper investigates the performance of IGS  in the  $K$-user MIMO IC.  
To the best of our knowledge, this is the first work to study IGS in the  $K$-user MIMO IC with HWI including I/Q imbalance. 
We employ the HWI model in \cite{javed2019multiple} and assume non-ideal transceivers, which generate an additive proper white Gaussian noise. 
Moreover, we assume that the upconversion (at the transmitter side) and/or the downconversion (at the receiver side) chains suffer from I/Q imbalance, which makes the received signal a function of the widely linear transform of the transmitted signal and the aggregated noise. 
Thus, the aggregated noise becomes improper, which also motivates us to consider improper signals and improper signaling.

It is known that by making signals improper, we introduce structure into them  by correlating their real and imaginary parts \cite{schreier2010statistical,adali2011complex,adali2014optimization}, which can bring benefits to the performance of interference-limited systems. Indeed, on the one hand, interference can be mitigated or suppressed more effectively at the receiver side when it has more structure. On the other hand, the differential  entropy of a Gaussian signal decreases if the signal is more structured  (e.g., non-circular).   Improper signaling  schemes can improve the overall system performance when the gain we get by receiving an improper interference  at the non-intended receivers overcomes the rate loss caused by transmitting improper signals.
Additionally, 
IGS provides more optimization parameters and hence brings more flexibility than PGS schemes to the design of interference-limited wireless communication systems. This feature can be exploited to improve the system performance. 
Note that IGS includes PGS as a special case, where the complementary variances are zero. Hence, 
PGS never outperforms the optimal IGS  scheme. In this paper, we investigate whether this flexibility in design leads to performance improvement for the $K$-user MIMO IC with HWI.

Throughout this paper, we consider two main performance metrics: the achievable rate and the energy efficiency (EE). 
The EE of a user is defined as the ratio of its achievable rate to its total power consumption \cite{zappone2015energy}. 
There are only a few works that study the energy-efficiency of IGS in SISO networks \cite{soleymani2019energy, Sole1909:Energy}, but the question of how these results translate to the $K$-user MIMO IC remains unanswered.
In this paper, we provide an answer and show that IGS 
 can be beneficial in MIMO systems in terms of  achievable rate and EE  as well. 
Interestingly, IGS provides more relative  gain in terms of achievable rate than in terms of energy efficiency. This is in agreement with our previous findings on energy-efficient IGS schemes for SISO systems \cite{soleymani2019energy, Sole1909:Energy}. 
For instance, in \cite{soleymani2019energy}, we derived the necessary and sufficient conditions for optimality of IGS in SISO UCR from an EE point of view and showed that these conditions are more restrictive  than those obtained when the achievable rate is used as performance metric instead. 
Moreover, although there are some benefits for IGS in terms of global EE (defined as the ratio of the total achievable rate of the network to the total power consumption of the network) for the $K$-user MIMO  IC, as we will show, these benefits may not be significant.
In other words, our numerical results suggest that IGS does not provide a significant gain in terms of global EE for the  $K$-user MIMO IC.

In order to analyze the performance of IGS, we consider different optimization problems such as 
achievable rate region, maximum sum-rate, energy-efficiency region and global energy efficiency.
To solve these non-convex problems, we first formulate a general optimization problem that encompasses all performance metrics under study and 
then, employ a majorization-minimization framework to obtain solutions for each problem. 
The main idea of this framework is based on the structure of the achievable rate or energy-efficiency functions in interference-limited systems when interference is treated as noise. 
Specifically, the achievable rate with TIN is a difference of two concave/convex functions. 
We exploit this feature and employ a majorization-minimization (MM) approach to derive a stationary point\footnote{In a constrained optimization problem, a stationary point of  satisfies the corresponding
Karush-Kuhn-Tucker (KKT) conditions \cite{lanckriet2009convergence}.} of every optimization problem in interference-limited systems with TIN in which the objective function and/or the constraints are linear functions of the rates.

Our numerical results show that IGS can improve the performance of the $K$-user MIMO IC with HWI. Additionally, the results show an interesting behavior of IGS as an interference-management technique. 
We observe that the IGS benefits decrease with the number of antennas either at the transmitter or receiver sides for a fixed number of users. 
This is due to the fact that interference can be managed more effectively by PGS when there are more available spatial dimensions, and consequently, IGS provides less gain as an interference-management technique  in MIMO systems. 
Note however that with I/Q imbalance, IGS always performs better than PGS even if its  benefits can be minor when the number of antennas increases.
We also observe that, for a fixed number of antennas, the benefits of IGS increase when the number of users grows. 
The reason is that the interference level increases with $K$, and the higher the level of interference, the better the performance for IGS.     
Additionally, our results show that the benefits of IGS increase with the imbalance level. The more improper the noise is, the more benefits can be achieved by IGS.

The main contributions of this paper can be summarized as follows:
\begin{itemize}
\item We propose HWI-aware IGS schemes for the $K$-user MIMO IC. We study two general performance metrics, i.e., the achievable rate and EE, and solve four different optimization problems. We derive a stationary point of the rate region, sum-rate maximization, EE region and global EE maximization problems.

\item To solve these optimization problems, we cast them as a general optimization problem  in which the objective function and/or the constraints are linear functions of the rates. 
We  then apply a unified framework to obtain a stationary point of these  optimization problems by majorization-minimization  algorithms.

\item Our results show that IGS can improve the performance of the  $K$-user MIMO IC with HWI 
in terms of  achievable rate and EE. We show that IGS provides more benefits in terms of achievable rate than in terms of energy efficiency.

\item Our numerical simulations suggest  that 
the  benefit of IGS schemes increases with $K$ and with the level of impairment  for a fixed number of antennas. 
However, IGS provides minor gains with respect to PGS when the number of antennas grows for a fixed number of users.

\end{itemize}
\begin{figure*}[t]
\centering
\includegraphics[width=0.9\textwidth]{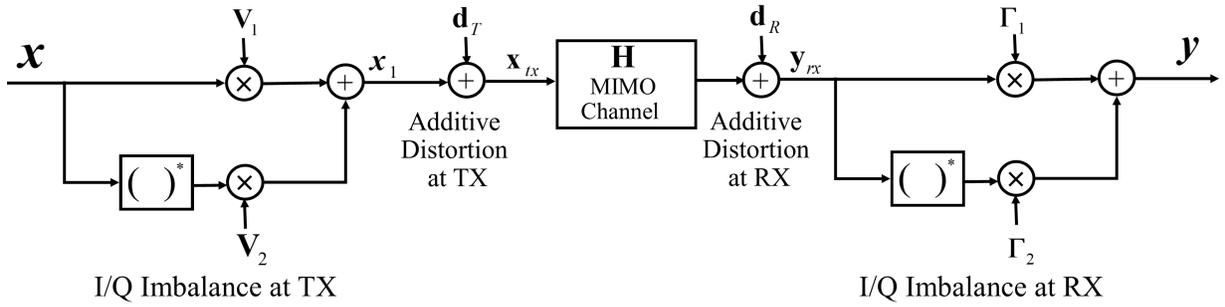}
\caption{The transceiver model of a point-to-point communications link with imperfect devices.}
\label{Fig1-1}
\end{figure*}
\subsection{Paper organization and notations}
This paper is organized as follows. In Section \ref{pre-sec-ii}, we present some background on improper random signals, as well as the system model considered in this work. 
We define the optimization framework for the MIMO IC with HWI at the transceivers in Section \ref{sec-framewrk}. 
We solve the corresponding optimization problems in Section \ref{sec-prob}. 
Finally, Section \ref{sec-v} provides some numerical examples  along with an extensive discussion of them.

{\it Notations and abbreviations}: 
In this paper, matrices are denoted by bold-faced upper case letters, bold-face  lower case letters denote column vectors, and scalars are denoted by light-face lower case letters.
Furthermore, $\text{Tr}(\mathbf{X})$ and $\det(\mathbf{X})$ denote, respectively, the trace and determinant of matrix $\mathbf{X}$. The notations $\mathfrak{R}\{\cdot\}$ and $\mathfrak{I}\{\cdot\}$ return, respectively, the real and imaginary part of $\{\cdot\}$ and can be applied to scalars, vectors and matrices. 
Additionally, $(\cdot)^H$, $(\cdot)^T$ and $(\cdot)^*$ denote, respectively, Hermitian, transpose and conjugate operations. We represent the $N\times N$ identity matrix by $\mathbf{I}_N$. 
Moreover, ${\cal CN}(0,1)$ denotes a proper complex Gaussian distribution with zero mean and unit variance, $\mathbf{x}\sim \mathcal{ CN}(\mathbf{0}, \mathbf{R} )$ denotes a proper complex Gaussian vector with zero mean and covariance $\mathbf{R}$.
Finally, we provide a list of the most frequently used abbreviations in Table \ref{table-1}. 
 \begin{table}\vspace{.2cm}
\caption{List of frequently used abbreviations.}\label{table-1}
\begin{tabular}{|l|l|}
	\hline
CCP&Convex-Concave Procedure\\
DCP& Difference of Convex Programming\\
EE & Energy Efficiency\\
GDA& Generalized Dinkelbach Algorithm\\
HWI& Hardware Impairment\\
IC& Interference Channel\\
IGS & Improper Gaussian Signaling\\
MIMO& Multiple-Input Multiple-Output\\
MM & Majorization Minimization\\
OFDM& Orthogonal Frequency Division Multiplexing\\
PGS& Proper Gaussian Signaling\\
QoS& Quality of Service\\
SNR & Signal to Noise Ratio\\
TIN& Treating Interference as Noise\\
UCR& Underlay Cognitive Radio\\
\hline
		\end{tabular}
\end{table} 

\section{Preliminaries and system model}\label{pre-sec-ii}
We provide some preliminaries on the real-decomposition method in Section \ref{sec-real-de}, and on improper signaling in Section \ref{real-IGS}. 
We then describe the HWI model in Section \ref{sec-hwi-mod}. We finally present the considered scenario in Section \ref{sec-rate}.

\subsection{Real decomposition of a complex system}\label{sec-real-de}
Consider the following point-to-point MIMO communication system
\begin{equation}
\mathbf{y}=\mathbf{H}\mathbf{x}+\mathbf{n},
\end{equation}
where $\mathbf{y}\in \mathbb{C}^{N_R\times 1}$, $\mathbf{x}\in \mathbb{C}^{N_T\times 1}$,  $\mathbf{n}\in \mathbb{C}^{N_R\times 1}$, and $\mathbf{H}\in \mathbb{C}^{N_R\times N_T}$ are, respectively, the received signal, transmitted signal, noise vector, and the channel matrix. 
The real decomposition model for the link is 
\begin{equation}
\left[\! \begin{array}{c}
\mathfrak{R}\{\mathbf{y}\} \\ \mathfrak{I}\{\mathbf{y}\} \end{array} \!\right]\!=\!
\left[\! \begin{array}{cc}
\mathfrak{R}\{\mathbf{H}\}& -\mathfrak{I}\{\mathbf{H}\} \\ \mathfrak{I}\{\mathbf{H}\}&
\mathfrak{R}\{\mathbf{H}\} \end{array} \right]\left[\! \begin{array}{c}
\mathfrak{R}\{\mathbf{x}\} \\ \mathfrak{I}\{\mathbf{x}\} \end{array}\! \right]\!+\!\left[\! \begin{array}{cc}\mathfrak{R}\{\mathbf{n}\} \\ \mathfrak{I}\{\mathbf{n}\} \end{array}\! \right]\!.
\end{equation}
Assume $\mathbf{n}$ is a random vector with Gaussian distribution as $\mathbf{n}\sim\mathcal{CN}(\mathbf{0},\mathbf{C}_{n})$.
The achievable rate of the system is \cite{cover2012elements}
\begin{equation}\label{rate-eq}
R_k=\frac{1}{2}\log_2\det\left(\underline{\mathbf{C}}_{n}+\underline{\mathbf{H}}\mathbf{P} \underline{\mathbf{H}}^T\right)-\frac{1}{2}\log_2\det\left(\underline{\mathbf{C}}_{n}\right),
\end{equation}
where $\underline{\mathbf{C}}_{n}$ is the covariance matrix of $[ \begin{array}{cc}\mathfrak{R}\{\mathbf{n}\}^T & \mathfrak{I}\{\mathbf{n}\}^T \end{array}]^T$, 
$\mathbf{P}$ is the covariance matrix of $[ \begin{array}{cc}\mathfrak{R}\{\mathbf{x}\}^T & \mathfrak{I}\{\mathbf{x}\}^T \end{array}]^T$, and $\underline{\mathbf{H}}$ is
\begin{equation}
\underline{\mathbf{H}}=\left[ \begin{array}{cc}
\mathfrak{R}\{\mathbf{H}\}& -\mathfrak{I}\{\mathbf{H}\} \\ \mathfrak{I}\{\mathbf{H}\}&
\mathfrak{R}\{\mathbf{H}\} \end{array} \right].
\end{equation}
\subsection{Preliminaries of IGS}\label{real-IGS}
A zero-mean complex Gaussian random variable $x$ with variance $p_t=\mathbb{E}\{|t|^2\}$ is called proper if $\mathbb{E}\{t^2\}=0$; otherwise, it is improper \cite{schreier2010statistical,adali2011complex}.  
When the variable $t$ is improper, its real and imaginary parts are not independent and identically distributed (i.i.d). 
We can extend the definition of improper scalar variables to vectors. A zero-mean complex Gaussian random vector $\mathbf{t}\in\mathbb{C}^{N\times 1}$ with covariance matrix $\mathbf{P}=\mathbb{E}\{\mathbf{t}\mathbf{t}^H\}$ is called proper if $\mathbb{E}\{\mathbf{t}\mathbf{t}^T\}=\mathbf{0}$; otherwise, it is improper \cite{schreier2010statistical,adali2011complex}.

To deal with improper signals, there are generally two approaches: augmented covariance matrix \cite{adali2011complex} and real decomposition method \cite{hellings2015block}.
In the augmented-covariance-matrix approach, complex-domain signals are considered, and the optimized variables are covariance and complementary covariance matrices. 
However, in the real-decomposition method, every variable is written in the real domain, and the optimization variable is the covariance matrix of the real decomposition of the signals. 
The main differences of these two approaches are in the structure of the optimization variables as well as in the corresponding optimization problems. That is, a complementary covariance matrix has to follow a specific structure for improper signals, while the real covariance matrices are required to be only  positive semi-definite. 
On the one hand, the use of the augmented-covariance-matrix approach can provide insights in some problems. For example, we might be able to derive some conditions for optimality of proper or improper signaling by considering complementary variances. 
On the other hand, depending on the scenario, the optimization over the real domain might be simpler. 
For instance, in MIMO systems, the achievable rate by IGS is a complicated function of the covariance and complementary covariance matrices (please refer to \cite[Eq. (10)]{zeng2013transmit} or \cite[Section III]{javed2019multiple}). 
  Then, using the standard complex formulation, it is not possible to express the rates as a concave/convex function or a difference of two concave functions in the optimization parameters, which makes the analysis intractable. This  is in contrast to the real-decomposition method by means of  which the rates can be written as a difference of two concave functions, as will be shown in Section \ref{sec-rate}. Therefore, the real-decomposition method is used to simplify the optimization problems throughout this paper.

It is worth emphasizing that, in the real decomposition model, an improper random vector can have any arbitrary symmetric and positive semi-definite covariance matrix. However, a proper Gaussian signal has a covariance matrix patterned as  \cite{schreier2010statistical}
\begin{align}\nonumber
\mathbf{P}
&=\mathbb{E}\left\{\left[\begin{array}{cc}\mathfrak{R}\{\mathbf{t}\}^T& \mathfrak{I}\{\mathbf{t}\}^T\end{array}\right]^T\left[\begin{array}{cc}\mathfrak{R}\{\mathbf{t}\}^T& \mathfrak{I}\{\mathbf{t}\}^T\end{array}\right]\right\}\\
&=\left[\begin{array}{cc}\mathbf{A}&\mathbf{B}\\ \mathbf{B}&\mathbf{A} \end{array}\right],\label{struc-pgs-cv}
\end{align}
where $\mathbf{A}\in \mathbb{R}^{N\times N}$ is symmetric and  positive semi-definite, and   $\mathbf{B}\in \mathbb{R}^{N\times N}$ is skew-symmetric, i.e., $\mathbf{B} = -\mathbf{B}^T$, which implies that its diagonal elements are zero.

\subsection{HWI model for MIMO systems}\label{sec-hwi-mod}
 
In this subsection, we present the HWI model for a MIMO system with $N_T$ transmitter antennas and $N_R$ receiver antennas (see Fig. \ref{Fig1-1}). 
We employ the non-ideal-hardware model in \cite{javed2019multiple} and assume that the transceivers suffer from I/Q imbalance and generate additive distortion noise. For the sake of completeness, we briefly present the model in \cite{javed2019multiple} in this subsection. 

The I/Q imbalance at the transmitter side is modeled as a widely linear transformation of the transmit signal $\mathbf{x}\in \mathbb{C}^{N_T\times 1}$ as 
\begin{equation}
\mathbf{x}_1=\mathbf{V}_1\mathbf{x}+\mathbf{V}_2\mathbf{x}^*,
\end{equation} 
where the matrices $\mathbf{V}_1\in \mathbb{C}^{N_T\times N_T}$ and
$\mathbf{V}_2 \in \mathbb{C}^{N_T\times N_T}$ capture the amplitude and rotational  imbalance and can be expressed as \cite{javed2019multiple}
\begin{align}
\mathbf{V}_1&=\frac{\mathbf{I}_{N_T}+\mathbf{A}_Te^{j \bm{\theta}_T}}{2},\\
\mathbf{V}_2&=\mathbf{I}_{N_T}-\mathbf{V}_1^*=\frac{\mathbf{I}_{N_T}-\mathbf{A}_Te^{-j \bm{\theta}_T}}{2}.
\end{align}
Moreover, the matrices $\mathbf{A}_T$ and $\bm{\theta}_T$ 
are diagonal and, respectively, reflect the amplitude and phase errors of each branch at the transmitter side \cite{javed2019multiple}. There are different methods to estimate the parameters of I/Q imbalance \cite{chung2010joint,gil2005joint,sung2009estimation,cai2011estimation}.
There is no I/Q imbalance if $\mathbf{A}_T=\mathbf{I}$ and $\bm{\theta}_T=\mathbf{0}$ or equivalently, $\mathbf{V}_1=\mathbf{I}$ and $\mathbf{V}_2=\mathbf{0}$.

We also assume that the transmitter is not perfect and may generate an additive proper Gaussian noise  in addition to the I/Q imbalance with probability distribution $\mathbf{d}_T\in \mathbb{C}^{N_T\times 1}\sim\mathcal{CN}(\mathbf{0},\mathbf{C}_T)$ \cite{javed2019multiple}. 
Hence, the transmitted signal is
\begin{equation}
\mathbf{x}_{tx}=\mathbf{x}_1+\mathbf{d}_T.
\end{equation}
The transmitted signal is delivered to the receiver over a MIMO channel with additive white Gaussian noise. 
Hence, the received signal 
is
\begin{equation}
\mathbf{y}_{rx}=\mathbf{H}\mathbf{x}_{tx}+\mathbf{d}_R,
\end{equation}
where the vector
 $\mathbf{d}_R\in \mathbb{C}^{N_R\times 1}\sim\mathcal{CN}(\mathbf{0},\mathbf{C}_R)$ accounts for the aggregate effect of the additive white Gaussian noise of the channel and the additive distortion of the receive devices. 
 The receiver can suffer from an I/Q imbalance similar to the transmitter.  
Thus, the received signal after I/Q imbalance is
\begin{equation}
\mathbf{y}=\mathbf{\Gamma}_1\mathbf{y}_{rx}+\mathbf{\Gamma}_2\mathbf{y}_{rx}^*,
\end{equation}
where the matrices 
$\mathbf{\Gamma}_1\in \mathbb{C}^{N_R\times N_R}$ and
$\mathbf{\Gamma}_2\in \mathbb{C}^{N_R\times N_R}$ are, respectively, given by
\begin{align}
\mathbf{\Gamma}_1&=\dfrac{\mathbf{I}_{N_R}+\mathbf{A}_Re^{j \bm{\theta}_R}}{2}\\
\mathbf{\Gamma}_2&=\mathbf{I}_{N_R}-\mathbf{\Gamma}_1^*=\frac{\mathbf{I}_{N_R}-\mathbf{A}_Re^{-j \bm{\theta}_R}}{2}.
\end{align}
Similar to $\mathbf{A}_T$ and $\bm{\theta}_T$ , the matrices $\mathbf{A}_R$ and $\bm{\theta}_R$ 
are diagonal and, respectively, reflect the amplitude and phase errors of each branch at the receiver side \cite{javed2019multiple}. 
The following lemmas present the aggregate effect of the impairments at the transmitter and receiver sides.
\begin{lemma}[\!\cite{javed2019multiple}]\label{lem-1}
The transceiver of a MIMO system with HWI can be modeled as
\begin{equation}\label{eq-hwi}
\mathbf{y}=\bar{\mathbf{H}}_1\mathbf{x}+\bar{\mathbf{H}}_2\mathbf{x}^*+\mathbf{z},
\end{equation}
where
\begin{align}
\label{eq-h1} \bar{\mathbf{H}}_1&=\mathbf{\Gamma}_1\mathbf{H}\mathbf{V}_1+\mathbf{\Gamma}_2\mathbf{H}^*\mathbf{V}_2^*\in \mathbb{C}^{N_R\times N_T},\\
\label{eq-h2}\bar{\mathbf{H}}_2&=\mathbf{\Gamma}_1\mathbf{H}\mathbf{V}_2+\mathbf{\Gamma}_2\mathbf{H}^*\mathbf{V}_1^*\in \mathbb{C}^{N_R\times N_T},\\
\mathbf{z}&=\mathbf{\Gamma}_1(\mathbf{H}\mathbf{d}_T+\mathbf{d}_R)+\mathbf{\Gamma}_2(\mathbf{H}\mathbf{d}_T+\mathbf{d}_R)^*\in \mathbb{C}^{N_R\times 1}.\label{eq-noise}
\end{align}
\end{lemma}
\begin{proof}
Please refer to \cite[Lemma 1]{javed2019multiple}.
\end{proof}
\begin{lemma}\label{lem:Rea}
The real decomposition  of the MIMO system with HWI in Lemma \ref{lem-1} is
\begin{equation}\label{eq-hwi-2}
\underline{\mathbf{y}}=
\tilde{\mathbf{H}}
\underline{\mathbf{x}}+\underline{\mathbf{z}},
\end{equation}
where $\underline{\mathbf{y}}=\left[ \begin{array}{cc}
\mathfrak{R}\{\mathbf{y}\}^T & \mathfrak{I}\{\mathbf{y}\}^T \end{array} \right]^T$,  
$\underline{\mathbf{x}}=\left[ \begin{array}{cc}
\mathfrak{R}\{\mathbf{x}\}^T & \mathfrak{I}\{\mathbf{x}\}^T \end{array} \right]^T$, and 
$\underline{\mathbf{z}}=\left[ \begin{array}{cc}
\mathfrak{R}\{\mathbf{z}\}^T & \mathfrak{I}\{\mathbf{z}\}^T \end{array} \right]^T$ are, respectively, the real decomposition of $\mathbf{y}$, $\mathbf{x}$, and $\mathbf{z}$ in \eqref{eq-hwi}. 
Moreover, $\tilde{\mathbf{H}}$ is
\begin{equation}
\tilde{\mathbf{H}}=\left[ \begin{array}{cc}
\mathfrak{R}\{\bar{\mathbf{H}}_1+\bar{\mathbf{H}}_2\}& -\mathfrak{I}\{\bar{\mathbf{H}}_1-\bar{\mathbf{H}}_2\} \\ \mathfrak{I}\{\bar{\mathbf{H}}_1+\bar{\mathbf{H}}_2\}&
\mathfrak{R}\{\bar{\mathbf{H}}_1-\bar{\mathbf{H}}_2\} \end{array} \right].
\end{equation}
The statistics of the vector $\underline{\mathbf{z}}\in \mathbb{R}^{2N_R\times 1}$ are $\mathbb{E}\{\underline{\mathbf{z}}\}=\mathbf{0}$, and
\begin{align}\label{noise-var}
\mathbb{E}\{\underline{\mathbf{z}}\,\underline{\mathbf{z}}^T\}&=\underline{\mathbf{C}}_z= \underline{\mathbf{\Gamma}}
\mathbf{C}_d
\underline{\mathbf{\Gamma}}^T,
\end{align}
where
$\mathbf{C}_d=\underline{\mathbf{H}}\,\underline{\mathbf{C}}_T\underline{\mathbf{H}}^T+\underline{\mathbf{C}}_R,$
 and
 \begin{equation}\label{gamma-eq}
\underline{\mathbf{\Gamma}}\triangleq \left[ \begin{array}{cc}
\mathfrak{R}\{\mathbf{\Gamma}_1+\mathbf{\Gamma}_2\}& -\mathfrak{I}\{\mathbf{\Gamma}_1-\mathbf{\Gamma}_2\} \\ \mathfrak{I}\{\mathbf{\Gamma}_1+\mathbf{\Gamma}_2\}&
\mathfrak{R}\{\mathbf{\Gamma}_1+\mathbf{\Gamma}_2\} \end{array} \right].
\end{equation}
 Additionally, $\underline{\mathbf{H}}$, $\underline{\mathbf{C}}_T$, and $\underline{\mathbf{C}}_R$ are, respectively, the real decomposition of $\mathbf{H}$, $\mathbf{C}_T$, and $\mathbf{C}_R$. 
For example, if $\mathbf{C}_T=\sigma^2\mathbf{I}_{N_T}$, then $\underline{\mathbf{C}}_T=\frac{1}{2}\sigma^2\mathbf{I}_{2N_T}$.
\end{lemma}
\begin{proof} 
We can easily construct the real decomposition model in \eqref{eq-hwi-2} from the complex model in  \eqref{eq-hwi}. Now we would like to derive the statistics of $\underline{\mathbf{z}}$ in  \eqref{eq-hwi-2}. To this end, we first write the real decomposition of $\mathbf{z}$ in \eqref{eq-noise} as
\begin{align}\nonumber
\left[ \begin{array}{c}
\mathfrak{R}\{\mathbf{z}\} \\ \mathfrak{I}\{\mathbf{z}\} \end{array} \right]&=
\left[ \begin{array}{cc}
\mathfrak{R}\{\mathbf{\Gamma}_1+\mathbf{\Gamma}_2\}& -\mathfrak{I}\{\mathbf{\Gamma}_1-\mathbf{\Gamma}_2\} \\ \mathfrak{I}\{\mathbf{\Gamma}_1+\mathbf{\Gamma}_2\}&
\mathfrak{R}\{\mathbf{\Gamma}_1+\mathbf{\Gamma}_2\} \end{array} \right]\\
&\times
\left[ \begin{array}{c}
\mathfrak{R}\{\mathbf{H}\mathbf{d}_T+\mathbf{d}_R\} \\ \mathfrak{I}\{\mathbf{H}\mathbf{d}_T+\mathbf{d}_R\} \end{array} \right],
\end{align}
which can be represented as $\underline{\mathbf{z}}=\underline{\mathbf{\Gamma}}\left(\underline{\mathbf{H}}\,\underline{\mathbf{d}}_T+\underline{\mathbf{d}}_R\right)$, where $\underline{\mathbf{d}}_T$ and $\underline{\mathbf{d}}_R$ are, respectively, the real decomposition of $\mathbf{d}_T$ and $\mathbf{d}_R$.
The average of $\underline{\mathbf{z}}$ is simply $\mathbf{0}$, since $\mathbf{d}_R$ and $\mathbf{d}_T$ are  zero-mean random vectors. 
Furthermore, the covariance matrix of $\underline{\mathbf{z}}$ can be derived as \eqref{noise-var}.

\end{proof}

\subsection{System model}\label{sec-rate}
\begin{figure}[t]
\centering
\includegraphics{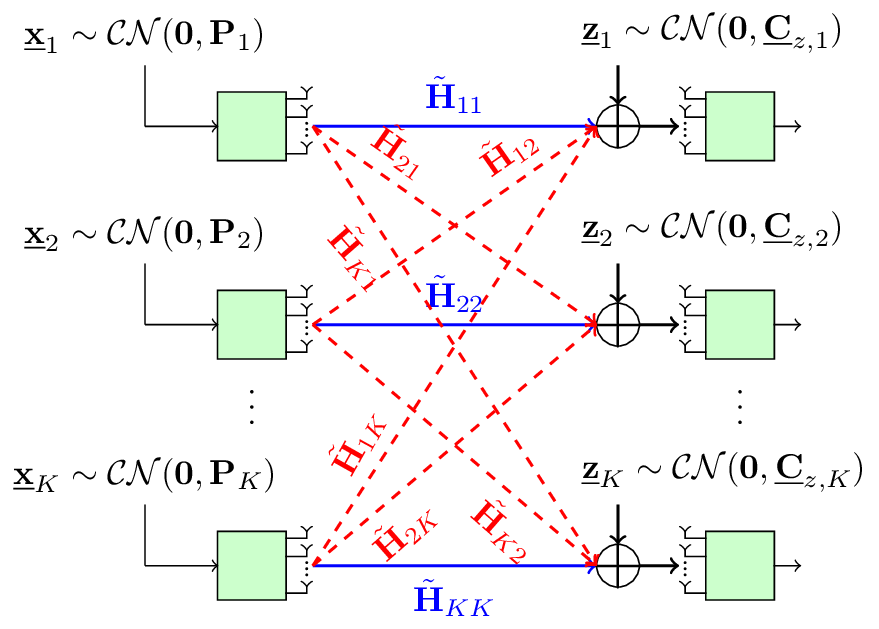}
\caption{The equivalent real-decomposition channel model for the  $K$-user MIMO IC.}
\label{Fig2}
\end{figure}
We consider a  $K$-user MIMO IC with imperfect transceivers, as shown in Fig. \ref{Fig2}. 
Without loss of generality, we assume that the transceivers have the same number of antennas and produce a noise with the same statistics to simplify the notation and the expressions. Obviously, it is very straightforward to extend this model to the most general case with asymmetric devices.
According to Lemma 2, the real decomposition of the received signal at the receiver of user $k$ is
\begin{equation}
\underline{\mathbf{y}}_k=\sum_{i=1}^{K}\tilde{\mathbf{H}}_{ki}\underline{\mathbf{x}}_i+\underline{\mathbf{z}}_k
\end{equation}
where $\underline{\mathbf{x}}_i$ is the real decomposition of the transmitted signal of user $i$, and
\begin{equation}
\tilde{\mathbf{H}}_{ki}=\left[ \begin{array}{cc}
\mathfrak{R}\{\bar{\mathbf{H}}_{1,ki}+\bar{\mathbf{H}}_{2,ki}\}& -\mathfrak{I}\{\bar{\mathbf{H}}_{1,ki}-\bar{\mathbf{H}}_{2,ki}\} \\ \mathfrak{I}\{\bar{\mathbf{H}}_{1,ki}+\bar{\mathbf{H}}_{2,ki}\}&
\mathfrak{R}\{\bar{\mathbf{H}}_{1,ki}-\bar{\mathbf{H}}_{2,ki}\} \end{array} \right],
\end{equation}
where $\bar{\mathbf{H}}_{1,ki}$ and $\bar{\mathbf{H}}_{2,ki}$ can be derived, respectively, by replacing $\mathbf{H}_{ki}$ in \eqref{eq-h1} and \eqref{eq-h2}. 
Note that $\mathbf{H}_{ki}$ is the channel matrix for the link between transmitter $i$ and receiver $k$.
Moreover, $\underline{\mathbf{z}}_k$ is the real decomposition of the noise vector $\mathbf{z}_k$, which is given by
\begin{equation}
\mathbf{z}_k\!=\!\mathbf{\Gamma}_1\!\left(\sum_{i=1}^{K}\mathbf{H}_{ki}\mathbf{d}_{T,i}\!+\mathbf{d}_{R,k}\!\!\right)+\mathbf{\Gamma}_2\!\left(\sum_{i=1}^{K}\mathbf{H}_{ki}\mathbf{d}_{T,i}\!+\mathbf{d}_{R,k}\!\!\right)^*\!\!.
\end{equation}
According to Lemma \ref{lem:Rea}, the covariance matrix of $\underline{\mathbf{z}}_k$ is
\begin{equation}
\underline{\mathbf{C}}_{z,k}=\underline{\mathbf{\Gamma}}
\left(\sum_{i=1}^K\underline{\mathbf{H}}_{ki}\underline{\mathbf{C}}_{T}\underline{\mathbf{H}}_{ki}^T+\underline{\mathbf{C}}_{R}
\right)\underline{\mathbf{\Gamma}}^T,
\end{equation}
where $\underline{\mathbf{H}}_{ki}$ is the real decomposition of $\mathbf{H}_{ki}$, and $\underline{\mathbf{\Gamma}}$ is given by \eqref{gamma-eq}.
Treating interference as noise,
we can derive the rate of user $k\in\{1,2,...,K\}$ as \cite{cover2012elements, schreier2010statistical,you2020spectral}
\begin{align}\nonumber
R_k&=
\underbrace{\frac{1}{2}\log_2\det\left(\underline{\mathbf{C}}_{z,k}+\sum_{i=1}^{K}\tilde{\mathbf{H}}_{ki}\mathbf{P}_i \tilde{\mathbf{H}}_{ki}^T\right)}
_{\triangleq\,r_{k,1}}\\
&-
\underbrace{\frac{1}{2}\log_2\det\left(\underline{\mathbf{C}}_{z,k}+\sum_{i=1,i\neq k}^{K}\tilde{\mathbf{H}}_{ki}\mathbf{P}_i \tilde{\mathbf{H}}_{ki}^T\right)}_{\triangleq\,r_{k,2}}.
\label{rate-eq}
\end{align}
As can be observed through \eqref{rate-eq}, the rate of user $k$ is a difference of two concave functions, 
i.e., $R_k=r_{k,1}-r_{k,2}$, where $r_{k,1}$ and $r_{k,2}$ are concave. 
This feature allows us to employ MM and convex/concave procedure (CCP) for optimization problems in which the objective and/or constraints are linear functions of the rates as will be shown in Section \ref{sec-framewrk} and Section \ref{sec-prob}. 

\section{Optimization framework for MIMO systems based on MM}\label{sec-framewrk}
In this section, we present a framework based on MM to solve a family of optimization problems in which either the objective function and/or the constraints are linear functions of the rates.
In this approach, we exploit the fact that the rate is a difference of two concave functions and solve the corresponding optimization problem iteratively. 
To this end, we apply the CCP to the rates and approximate the convex part of the rates, $-r_{k,2}$, by a linear function through a first-order Taylor expansion. 

This framework can be applied to both IGS and PGS schemes. The only difference of IGS and PGS schemes in the framework is the feasibility set of the covariance matrices. 
As indicated in Section \ref{real-IGS}, an improper Gaussian random variable can have an arbitrary symmetric and positive semi-definite covariance matrix. 
Thus, the feasibility set of the covariance matrices of users $\{\mathbf{P}_k\}_{k=1}^K$ for IGS is 
\begin{equation}
\mathcal{P}_{IGS}= \left\{ \{\mathbf{P}_k\}_{k=1}^K :\text{Tr}(\mathbf{P}_k)\leq P_k, \,\, \mathbf{P}_k \succcurlyeq\mathbf{0},\forall k \right\},
\end{equation}
 where $P_k$ is the power budget of user $k$.
It is in contrast with a proper Gaussian signal, which has a covariance matrix with the specific structure in \eqref{struc-pgs-cv}. In this case, the feasibility set is 
\begin{equation}
\mathcal{P}_{PGS}\!=\! \left\{ \{\mathbf{P}_k\}_{k=1}^K\!\! :\!\text{Tr}(\mathbf{P}_k)\leq P_k,  \mathbf{P}_k= \mathbf{P}_{x_k}, \mathbf{P}_k \succcurlyeq\mathbf{0},\forall k \right\}\!,
\end{equation}
where $\mathbf{P}_{x_k}$ has the structure in \eqref{struc-pgs-cv}. 
In order to include both IGS and PGS schemes in the derivations to follow, we denote the feasibility set of the covariance matrices as $\mathcal{P}$ hereafter.

Consider the following optimization problem
\begin{subequations}\label{ar-opt}
\begin{align}
 \underset{\{\mathbf{P}_k\}_{k=1}^K\in\mathcal{P}
 }{\max}\,\,\,\,\,\,\,\,  & 
  f_0\left(\left\{\mathbf{P}_k\right\}_{k=1}^K\right) &\\
 \,\,\,\,\,\,\,\,\,\,\,\, \,\, \text{s.t.}   \,\,\,\,\,\,\,\,\,\,&  f_i\left(\left\{\mathbf{P}_k\right\}_{k=1}^K\right)\geq0,&i=1,2,...,I.
 \end{align}
\end{subequations}
If $f_0(\cdot)$ and $f_i(\cdot)$ for $i=1,2,...,I$ are concave, the optimization problem \eqref{ar-opt} is known to be convex\footnote{In \cite{aubry2018new}, it is defined as a concave optimization problem. However, we call it convex  since such an optimization problem is also widely known as a convex optimization problem \cite{boyd2004convex}.} and can be solved in polynomial time. 
If $f_0(\cdot)$ and $f_i(\cdot)$ for $i=1,2,...,I$ are neither concave nor pseudo-concave, it is not straightforward to derive the global optimal solution of \eqref{ar-opt} in  polynomial time \cite{yang2017unified,sun2017majorization,boyd2004convex,aubry2018new}. 
A way to solve non-convex optimization problems is to employ iterative optimization algorithms such as MM. 
The MM algorithm consists of two steps at each iteration: majorization and minimization. 
In the majorization step, the functions $f_0(\cdot)$ and $f_i(\cdot)$ for $i=1,2,...,I$ are approximated by surrogate functions. 
Then, the corresponding surrogate problem is solved in the minimization step.
In the following lemma, we present convergence conditions of MM iterative algorithms.
\begin{lemma}[\!\cite{aubry2018new}]\label{lem-sur}
Let us define $\tilde{f}_i^{(l)}(\cdot)$ for $l\in\mathbb{N}$ as surrogate functions of $f_i(\cdot)$ for $i=0,1,2,...,I$ such that the following conditions are fulfilled:
\begin{itemize}
\item $\tilde{f}_i^{(l)}\left(\{\mathbf{P}_k^{(l)}\}_{k=1}^{K}\right)=f_i\left(\{\mathbf{P}_k^{(l)}\}_{k=1}^{K}\right)$ for $i=0,1,2,\cdots,I$.
\item $\frac{\partial \tilde{f}_i^{(l)}\left(\{\mathbf{P}_k^{(l)}\}_{k=1}^{K}\right)}{\partial \mathbf{P}_k}
=\frac{\partial f_{i}\left(\{\mathbf{P}_k^{(l)}\}_{k=1}^{K}\right)}{\partial \mathbf{P}_k}$ for $i=0,1,2,\cdots,I$ and $k=1,2,\cdots,K$.
\item $\tilde{f}_i^{(l)}(\cdot)\leq f_i(\cdot)$ for $i=0,1,2,\cdots,I$ for all feasible $\left\{\mathbf{P}_k\right\}_{k=1}^K$,
\end{itemize}
where $\{\mathbf{P}_k^{(l)}\}_{k=1}^K$ is the initial point at the $l$-th iteration of the MM algorithm, which is obtained by solving 
\begin{subequations}\label{ar-opt-2}
\begin{align}
 \underset{\{\mathbf{P}_k\}_{k=1}^K\in\mathcal{P}
 }{\max}\,\,\,\,\,\,\,\,  & 
  \tilde{f}_0^{(l-1)}\left(\left\{\mathbf{P}_k\right\}_{k=1}^K\right) \\
 \,\,\,\,\,\,\,\, \,\,\,\,\,\,\, \text{{\em s.t.}}  \,\,\,\,\,\,\,\,\,&  \tilde{f}_i^{(l-1)}\left(\left\{\mathbf{P}_k\right\}_{k=1}^K\right)\geq0,&\forall i.
 \end{align}
\end{subequations}
Then, the sequence of $\{\mathbf{P}_k^{(l)}\}_{k=1}^K$ converges to a stationary point of \eqref{ar-opt}. 
\end{lemma}
\begin{remark}
The surrogate optimization problem \eqref{ar-opt-2} is not necessarily  convex; however, we can obtain the  global optimal solution of \eqref{ar-opt-2} much more easily than \eqref{ar-opt}. \end{remark}
Note that finding surrogate functions depends on the structure of the objective and constraint functions. 
In general, there might be different approaches to obtain a surrogate function (see, e.g., \cite{sun2017majorization}). 
As indicated in Section \ref{sec-rate}, the rate of each user is a difference of two concave functions, which allows us to apply CCP to obtain a suitable surrogate function. That is, we approximate the convex part of the rate expressions in \eqref{rate-eq} by its first-order Taylor series expansion, which is a linear function.
By MM and CCP, we are able to obtain a stationary point of optimization problems in which either the objective or constraint functions are linear functions of the rates of the users, as will be discussed in Section \ref{sec-prob}. 
In the following lemmas, we present the surrogate functions for the rates.
\begin{lemma}
\label{lem-R-pre}
Using CCP, we can obtain an affine upper bound for $\log\det(\mathbf{Q})$ as
\begin{equation}\label{alaki}
\log_2\det(\mathbf{Q})\! \leq\! \log_2\det\left(\!\mathbf{Q}^{(l)}\right)\! +\frac{1}{\ln 2} \text{\em{Tr}}((\mathbf{Q}^{(l)})^{-1}(\mathbf{Q}-\mathbf{Q}^{(l)})), 
\end{equation}
where $\mathbf{Q}^{(l)}$ is any feasible fixed point.
\end{lemma}
\begin{proof}
A concave function can be majorized by an affine function if the two functions have the same value and the same derivative in a point \cite{sun2017majorization}. The logarithmic function is concave. Furthermore, the left-hand and the right-hand sides of \eqref{alaki} hold these conditions at $\mathbf{Q}=\mathbf{Q}^{(l)}$. Thus, the upper-bound in \eqref{alaki} holds for all feasible $\mathbf{Q}$.
Note that the derivative of $\log\det(\mathbf{Q})$ with respect to $\mathbf{Q}$ is $\mathbf{Q}^{-1}$. 
\end{proof}
 \begin{lemma} \label{lem-R}
A concave approximation of the rates in \eqref{rate-eq} can be obtained by CCP as 
\begin{align}\nonumber
R_k\geq \tilde{R}_k^{(l)} &=r_{k,1}-r_{k,2}\left(\{\mathbf{P}_i^{(l)}\}_{i=1}^{K}\right)\\
&-\text{{\em Tr}}\left(\sum_{i=1,i\neq k}^{K}\frac{\partial r_{k,2}\left(\{\mathbf{P}_i^{(l)}\}_{i=1}^{K}\right)^T}{\partial \mathbf{P}_i}(\mathbf{P}_i-\mathbf{P}_i^{(l)})\right)
\end{align}
where $r_{k,1}$ and $r_{k,2}$ are, respectively, the concave and convex parts of $R_k$ in \eqref{rate-eq}. 
Moreover, $\frac{\partial r_{k,2}\left(\{\mathbf{P}_i^{(l)}\}_{i=1}^{K}\right)}{\partial \mathbf{P}_i}$ is the derivative of $r_{k,2}$ with respect to $\mathbf{P}_i$ at the previous iteration as
\begin{multline}
\frac{\partial r_{k,2}\left(\{\mathbf{P}_i^{(l)}\}_{i=1}^{K}\right)}{\partial \mathbf{P}_i}
=\\\frac{1}{\ln 2}
\tilde{\mathbf{H}}_{ki}^T
\left(\underline{\mathbf{C}}_{z,k}+\sum_{i=1,i\neq k}^{K}\tilde{\mathbf{H}}_{ki}\mathbf{P}_i^{(l)} \tilde{\mathbf{H}}_{ki}^T\right)^{-1}\tilde{\mathbf{H}}_{ki}.
\end{multline}
Note that $r_{k,2}\left(\{\mathbf{P}_i^{(l)}\}_{i=1}^{K}\right)$ is constant and is given by $r_{k,2}$ at the previous step. Additionally, $R_k$ and $\tilde{R}_k^{(l)}$ fulfill the conditions in Lemma \ref{lem-sur}.
\end{lemma}

\section{Optimization problems}\label{sec-prob}
In this section, we obtain a stationary point of the rate region, sum-rate maximization, EE region and global EE maximization problems by the framework described in Section \ref{sec-framewrk}. 

\subsection{Achievable rate region}

The rate region for the  $K$-user MIMO  IC with HWI can be derived by the rate-profile technique  as \cite{zeng2013transmit}
\begin{align}\label{rateregion}
 \underset{R,\{\mathbf{P}_k\}_{k=1}^K\in\mathcal{P}
 }{\max}\,\,\,\,\,\,\,\,  & 
  R &
 \,\,\,\,\,\,\,\, \,\,\,\,\,\,\, \text{s.t.}  \,\,\,\,\,\,\,\,\,&  R_{k}\geq\alpha_k R,&\forall k,
 \end{align}
where $\alpha_k\geq 0$ for $k=1,2,\cdots,K$ are given constants, and $\sum_{k=1}^K\alpha_k=1$. The boundary of the achievable rate region can be derived by solving \eqref{rateregion} for different values of the $\alpha_k$s.  
The optimization problem \eqref{rateregion} is not convex; however, we can obtain its stationary
point by the framework proposed in Section \ref{sec-framewrk}. 
That is, we solve \eqref{rateregion} iteratively, and in each iteration, we employ the surrogate function in Lemma \ref{lem-R} for the rates. 
Since the corresponding surrogate optimization problem is convex, we can efficiently derive the global optimal solution of each surrogate optimization problem and consequently, obtain a stationary point of \eqref{rateregion}.

\subsection{Maximizing sum-rate}
The sum-rate of the  $K$-user MIMO IC with HWI can be obtained by solving 
\begin{subequations}
\begin{align}
 \underset{\{\mathbf{P}_k\}_{k=1}^K\in\mathcal{P}
 }{\max}\,\,\,\,\,\,\,\,  & 
  \sum_{k=1}^KR_k &\\
 \label{qos-r} \text{s.t.}  \,\,\,\,\,\,\,\, \,\,\,\,\,\,\,\,\,\,\,\,\,\,\,\,&  R_{k}\geq R_{\text{th},k},&\forall k,
 \end{align}
\label{rate-sum}
\end{subequations}
\!\!where \eqref{qos-r} is the quality of service (QoS) constraint, and $R_{\text{th},k}$ is a given threshold for the rate of user $k$. 
Note that the $R_{\text{th},k}$s have to be set to make \eqref{rate-sum} feasible. 
Similar to \eqref{rateregion}, we can solve \eqref{rate-sum} by the framework in Section \ref{sec-framewrk} and obtain its stationary point. 
Note that each surrogate optimization problem is convex, which can be solved efficiently.

\subsection{Energy-efficiency region}
Now we consider the EE of the  $K$-user MIMO IC with HWI. 
The EE of user $k$ is defined as the ratio of its achievable rate to its power consumption \cite{zappone2015energy}
\begin{equation}
E_k=\frac{R_k}{\eta_k \text{Tr}(\mathbf{P}_k)+P_{c,k}}, \hspace{1cm} \text{(bits/Joule)}
\end{equation} 
where $\eta_k^{-1}$, and $P_{c,k}$ are, respectively, the power transmission efficiency of user $k$, and the constant power consumption of the $k$-th transceiver. 
The EE function is a linear function of the rates, which allows us to apply the framework in Section \ref{sec-framewrk} to optimize the EE. 
 EE function has a fractional structure, which makes its optimization more difficult than the rate analysis, as will be discussed in the following. 

The EE region of the  $K$-user MIMO IC with HWI can be derived by solving \cite{zappone2015energy}
\begin{subequations}
\begin{align}
 \underset{E,\{\mathbf{P}_k\}_{k=1}^K\in\mathcal{P}
 }{\max}\,\,\,\,\,\,\,\,  & 
  E &\\
 \label{5bbb} \text{s.t.}  \,\,\,\,\,\,\,\, \,\,\,\,\,\,\,\,\,\,\,\,\,\,\,\,&  E_{k}=\frac{R_k}{\eta_k \text{Tr}(\mathbf{P}_k)+P_{c,k}}\geq\alpha_k E,&\forall k,\\
\label{qos-eee}&  R_{k}\geq R_{\text{th},k},\hspace{1cm}\forall k,
 \end{align}
\label{ee-region}
\end{subequations}
\!where 
the constraint \eqref{qos-eee} is the QoS constraint, similar to \eqref{qos-r},
and $R_{\text{th},k}$ has to be chosen such that the feasible set of parameters is not empty. 
Similar to \eqref{rateregion}, the boundary of the EE region can be derived by solving \eqref{ee-region} for all possible $\alpha_k$s. 
Since $E_k$ is a linear function of $R_k$, we can apply the framework in Section \ref{sec-framewrk} to derive a stationary point of  \eqref{ee-region}. The surrogate optimization problem at the $l$-th iteration is
\begin{subequations}
\begin{align}
 \underset{E,\{\mathbf{P}_k\}_{k=1}^K\in\mathcal{P}
 }{\max}\,\,\,\,\,\,\,\,  & 
  E &\\
 \label{5b} \text{s.t.}  \,\,\,\,\,\,\,\,\,\,\,\,\,\,&  \tilde{E}_{k}^{(l)}=\frac{\tilde{R}_k^{(l)}}{\eta_k \text{Tr}(\mathbf{P}_k)+P_{c,k}}\geq\alpha_k E,&\forall k,\\
\label{qos-e}&  \tilde{R}^{(l)}_{k}\geq R_{\text{th},k},\hspace{1cm}\forall k.
 \end{align}
\label{ee-region-2}
\end{subequations}
Note that we can rewrite \eqref{ee-region-2} as a maximin fractional optimization problem by removing $E$ as
\begin{equation}\label{ee-region-3}
\underset{\{\mathbf{P}_k\}_{k=1}^K\in\mathcal{P}
 }{\max}\,\,\,\, \underset{1\leq k \leq K}\min \left\{\frac{\tilde{E}_{k}^{(l)}}{\alpha_k} \right\} 
 \,\,\,\,\,\, \text{s.t.} \,\,\,\,\, \tilde{R}^{(l)}_{k}\geq R_{\text{th},k},\,\,\,\,\forall k.
\end{equation}
The optimization problem \eqref{ee-region-2} (or equivalently \eqref{ee-region-3}) is not convex; however, its global optimum can be derived by employing the generalized Dinkelbach algorithm (GDA). 
The GDA is a powerful tool to solve maximin fractional optimization problems and is presented in the following Lemma.
\begin{lemma} \label{lem-gda}
Consider the following fractional optimization problem
\begin{equation}\label{eq-lem-gda}
\underset{\{\mathbf{X}\}\in\mathcal{X}
 }{\max}\,\,\,\,\,\,\underset{1\leq k\leq K}{\min}\left\{\frac{v_k\left(\mathbf{X}\right)}{u_k\left(\mathbf{X}\right)}\right\},
\end{equation}
where $v_k\left(\cdot\right)$ is a concave function in $\mathbf{X}$, $u_k\left(\cdot\right)$ is a convex function in $\mathbf{X}$,  
and $\mathcal{X}$ is a compact set. 
The global optimal solution of \eqref{eq-lem-gda} can be derived, iteratively, by the GDA, i.e., by solving
\begin{align}\label{eq-lem-gda-2}
\underset{t,\{\mathbf{X}\}\in\mathcal{X}
 }{\max}\,\,\,\,& t&
  \,\,\,\,\text{{\em s.t.}}  \,\,\,\,\,\,\,& 
 v_k\left(\mathbf{X}\right)-\mu^{(m)}u_k\left(\mathbf{X}\right)\geq t,&\forall k,
\end{align}
where $\mu^{(m)}$ is 
\begin{equation}
\mu^{(m)}=\underset{1\leq k\leq K}{\min}\left\{\frac{v_k\left(\mathbf{X}^{(m-1)}\right)}{u_k\left(\mathbf{X}^{(m-1)}\right)}\right\},
\end{equation}
where $\mathbf{X}^{(m-1)}$ is the solution of \eqref{eq-lem-gda-2} at the $(m-1)$th iteration. 
The GDA converges to the global optimum of \eqref{eq-lem-gda} linearly.
\end{lemma}
\begin{proof} Please refer to \cite{zappone2015energy,soleymani2019improper,crouzeix1991algorithms}.
\end{proof}
\begin{remark} 
In order to apply the GDA, it is not necessary that $v_k$ and $-u_k$ are concave in $\mathbf{X}$. However, the GDA is ensured to obtain the global optimum if  $v_k$ and $-u_k$ are concave in $\mathbf{X}$.
\end{remark}
Applying the GDA to \eqref{ee-region-2}, we have 
\begin{subequations}
\begin{align}
 \underset{E,\{\mathbf{P}_k\}_{k=1}^K\in\mathcal{P}
 }{\max}\,  & 
  E &\\
 \label{5b} \text{s.t.}\,\,\,\,\,\,\,\,\,\,&  \tilde{R}_k^{(l)}-\mu^{(m)}\left(\eta_k \text{Tr}(\mathbf{P}_k)+P_{c,k}\right)\geq\alpha_k E,&\forall k,\\
\label{qos-e}&  \tilde{R}^{(l)}_{k}\geq R_{\text{th},k},\hspace{1cm}\forall k,
 \end{align}
\label{ee-region-gda}
\end{subequations}
where
\begin{equation}
\mu^{(m)}=\underset{1\leq k\leq K}{\min}\left\{\frac{E_k\left(\{\mathbf{P}_i^{(l,m-1)}\}_{i=1}^K\right)}{\alpha_k}\right\},
\end{equation}
where $\{\mathbf{P}_i^{(l,m-1)}\}_{i=1}^K$ is the solution of \eqref{ee-region-gda} at the $(m-1)$th iteration.
Note that the GDA converges to the global optimum of \eqref{ee-region-2} linearly, and the whole algorithm converges to a stationary point of \eqref{ee-region}.

\subsection{Global energy-efficiency}
In this subsection, we consider the global EE of the $K$-user MIMO IC with HWI, which can be cast as \cite{zappone2015energy}
\begin{subequations}\label{ee-global}
\begin{align}
 \underset{\{\mathbf{P}_k\}_{k=1}^K\in\mathcal{P}
 }{\max}\,\,\,\,\,\,\,\,  & 
  G=\frac{\sum_{k=1}^KR_k}{\sum_{k=1}^K\left(\eta_k \text{Tr}(\mathbf{P}_k)+P_{c,k}\right)} \\
 \,\,\,\,\,\,\,\,\,\,\,\, \text{s.t.}   \,\,\,\,\,\,\,\,\,\,\,\,&  R_{k}\geq R_{\text{th},k},\hspace{1cm}\forall k.
 \end{align}
\end{subequations}
Similar to \eqref{ee-region}, since the EE is a linear function of the rates, we can apply the framework in Section \ref{sec-framewrk} to obtain a stationary point of \eqref{ee-global}. Thus, the surrogate optimization problem at the $l$-th iteration is 
\begin{subequations}
\begin{align}\label{ee-global-2}
 \underset{\{\mathbf{P}_k\}_{k=1}^K\in\mathcal{P}
 }{\max}\,\,\,\,  & 
  \tilde{G}=\frac{\sum_{k=1}^K\tilde{R}_k}{\sum_{k=1}^K\left(\eta_k \text{Tr}(\mathbf{P}_k)+P_{c,k}\right)} \\
  \text{s.t.}  \,\,\,\,\,\,\,\,\,\,\,\,&  \tilde{R}_{k}\geq R_{\text{th},k},\hspace{1cm}\forall k.
  \end{align}
\end{subequations}
Similar to \eqref{ee-region-2}, the optimization problem  is not convex; however, its global optimal solution can be derived by the Dinkelbach algorithm.
That is, we obtain $\mathbf{P}_i^{(l,m)}$ by solving
\begin{subequations}
\begin{align}
 \underset{\{\mathbf{P}_k\}_{k=1}^K\in\mathcal{P}
 }{\max}\,\,\,\,\,\,\,\,  & 
  \sum_{k=1}^K\tilde{R}_k-\mu^{(m)}\sum_{k=1}^K\left(\eta_k \text{Tr}(\mathbf{P}_k)+P_{c,k}\right) &\\
 \label{5b} \text{s.t.}  \,\,\,\,\,\,\,\, \,\,\,\,\,\,\,\,\,\,\,\,\,\,\,\,&  \tilde{R}_{k}\geq R_{\text{th},k},\hspace{1cm}\forall k,
  \end{align}
\label{ee-global-gda}
\end{subequations}
\!\!where $\mu^{(m)}=\tilde{G}\left(\{\mathbf{P}_i^{(l,m-1)}\}_{i=1}^K\right),$ in which $\{\mathbf{P}_i^{(l,m-1)}\}_{i=1}^K$ is the solution of \eqref{ee-global-gda} at the $(m-1)$th iteration. 
The global optimal solution of \eqref{ee-global-2} can be achieved by iteratively solving \eqref{ee-global-gda} and updating $\mu^{(m)}$ until a convergence metric is met.
Moreover, as indicated, the whole algorithm converges to a stationary point of \eqref{ee-global}.

\section{Numerical examples}\label{sec-v}

In this section, we provide some numerical examples. 
We employ Monte Carlo simulations and average the results over 100 independent  channel realizations. 
In each channel realization, the channel entries are drawn from a zero-mean complex proper Gaussian distribution with unit variance, i.e., $\mathcal{CN}(0,1)$.
For all simulations, the maximum number of the iterations of the MM algorithm is set to 40. 
We also consider $\mathbf{C}_T=\sigma^2_T\mathbf{I}_{N_T}$ and $\mathbf{C}_R=\sigma^2_R\mathbf{I}_{N_R}$ \cite{javed2019multiple}, or equivalently $\underline{\mathbf{C}}_T=\frac{1}{2}\sigma^2_T\mathbf{I}_{2N_T}$ and $\underline{\mathbf{C}}_R=\frac{1}{2}\sigma^2_R\mathbf{I}_{2N_R}$. 
In all simulations, we assume $\sigma^2_T=0.2$ and $\sigma^2_R=1$.
We assume that the I/Q imbalance by each antenna is the same. In other words, the matrices $\mathbf{A}_T=a_T\mathbf{I}_{N_T}$, $\bm{\theta}_T=\phi_T\mathbf{I}_{N_T}$,  $\mathbf{A}_R=a_R\mathbf{I}_{N_R}$, and $\bm{\theta}_R=\phi_R\mathbf{I}_{N_R}$ are scaled identity matrices. We consider $\phi_T=\phi_R=5$ degrees in all simulations.
We also define the signal-to-noise ratio (SNR) as the ratio of the power budget to $\sigma^2$, i.e., SNR$=\frac{P}{\sigma^2}$.

To the best of our knowledge, there are no other IGS algorithms in the literature that optimize EE or rate functions in the K-user MIMO IC.
Therefore, we compare our proposed algorithms for PGS and IGS with the PGS algorithm for ideal devices. 
The considered schemes in this section are as follows:
\begin{itemize}
\item {\bf IGS}: The IGS scheme.
\item {\bf PGS}: The PGS scheme.
\item {\bf I-PGS}: The PGS scheme for $K$-user IC without considering the I/Q imbalance in the design.
\end{itemize}

Note that the performance of MM algorithms depend on the initial point. In PGS and I-PGS algorithms, we start with a uniform power allocation $\mathbf{P}_k=\frac{P}{2N_T}\mathbf{I}_{2N_T}$ for $k=1,\cdots,K$ for optimization problems \eqref{rateregion} and \eqref{rate-sum}; and $\mathbf{P}_k=\frac{0.3P}{2N_T}\mathbf{I}_{2N_T}$ for $k=1,\cdots,K$ for optimization problems \eqref{ee-region} and \eqref{ee-global}. On the other hand, the IGS algorithm takes the solution of the PGS algorithm as an initial point.
\begin{figure}[t!]
    \centering
    \begin{subfigure}[t]{0.5\textwidth}
        \centering
\includegraphics{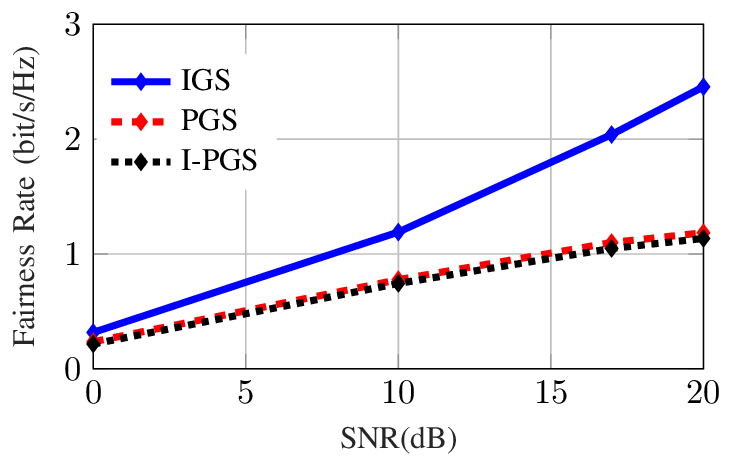}
        \caption{SISO.}
    \end{subfigure}%
    \\
    \begin{subfigure}[t]{0.5\textwidth}
        \centering
\includegraphics{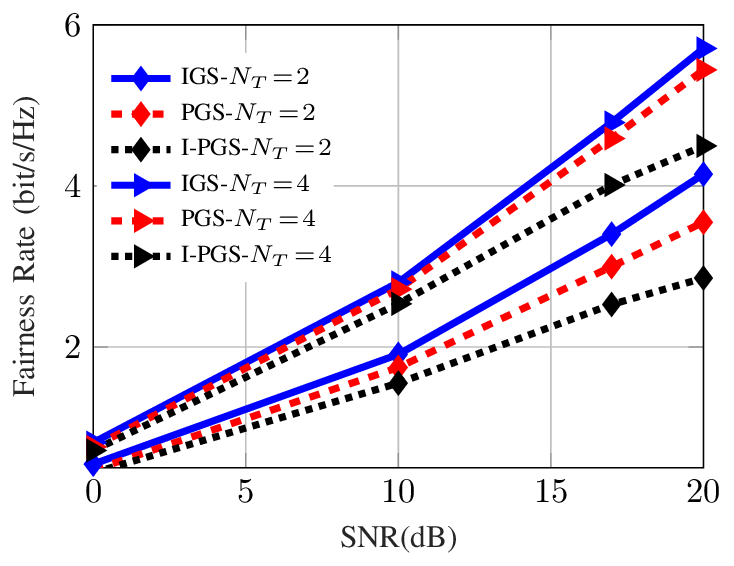}
        \caption{MISO.}
    \end{subfigure}
    \caption{Average fairness rate versus SNR for the 2-user SISO and MISO IC.}
	\label{Fig-r-simo}
\end{figure}
\subsection{Achievable rate region}
In this subsection, we consider a specific point of the rate region, which maximizes the minimum rate of users. 
The minimum rate of the $K$-user MIMO  IC is maximized for $\alpha_k=\frac{1}{K}$. 
This point is also referred as the maximin point or the fairness point. Hereafter, we call the maximin rate the fairness rate.  
We show the fairness rate of the 2-user SISO and MISO  IC for $a_T=0.6$ and different number of antennas at the transmitter side in Fig. \ref{Fig-r-simo}. 
As can be observed, there is a huge performance improvement by IGS in the 2-user SISO IC, especially at high SNR. 
However, the benefits of employing IGS become less substantial when  the number of antennas  increases. 
This is due to the fact that, by increasing the number of spatial resources (e.g. antennas) for a fixed number of users, the interference of PGS schemes can be managed more easily, and hence, IGS as an interference-management tool does not provide an additional significant gain. As indicated, this is in line with the results in \cite{soleymani2018improper}, in which it was shown that IGS might not provide significant benefits in OFDM UCR systems due to the existence of multiple parallel channels over which interference can be managed efficiently without resorting to IGS. 
Moreover, in \cite{hellings2018improper}, it was shown that the IGS does not provide any benefit in comparison to PGS with time sharing when the average power consumption is constrained instead of the  instantaneous power, which allows a more flexible power allocation. 
To sum up, the benefits of IGS decrease or even vanish when increasing the number of resources either by increasing the number of antennas or number of time slots, by time sharing, and/or the number of parallel channels by OFDM. 
Nevertheless, it is worth mentioning that, even with a large number of antennas, IGS and HWI-aware PGS outperform PGS, which is designed for ideal devices.  
Furthermore, IGS always performs not worse than PGS since IGS includes PGS  as a particular case. With I/Q imbalance, IGS performs better than PGS even in the MISO case with $N_T=4$. However, the benefits of IGS as an interference-management technique are not significant in this case.

\begin{figure}[t!]
    \centering
    \begin{subfigure}[t]{0.5\textwidth}
        \centering
\includegraphics{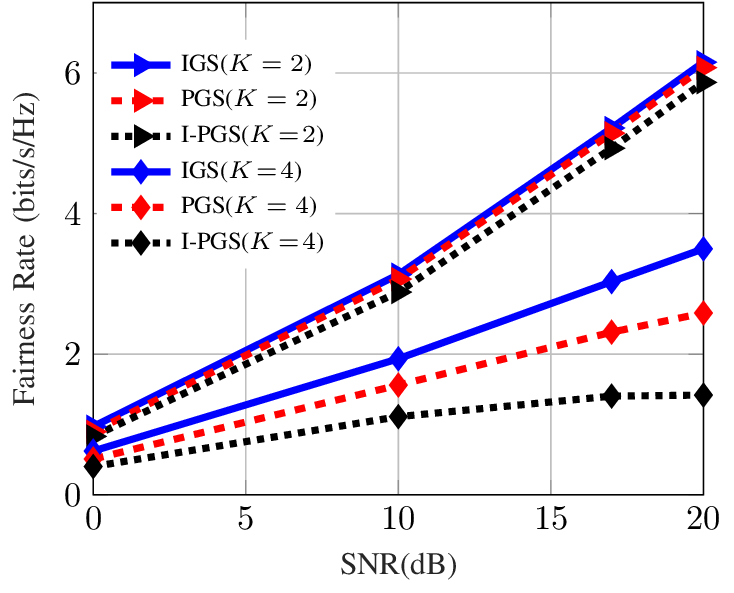}
        \caption{$K=2$ and $K=4$.}
    \end{subfigure}%
    \\ 
    \begin{subfigure}[t]{0.5\textwidth}
        \centering
           \includegraphics{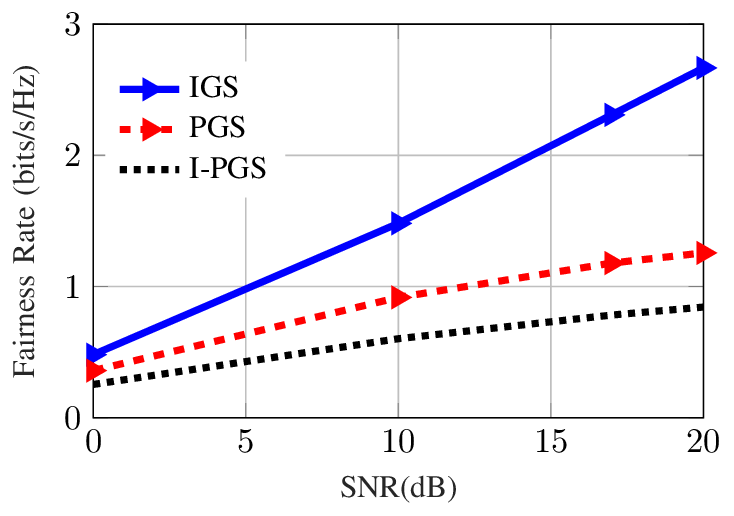}
        \caption{$K=6$.}
    \end{subfigure}
    \caption{Average fairness rate versus SNR for the $K$-user $2\times 2$ MIMO IC.}
	\label{Fig-r-snr}
\end{figure}
Figure \ref{Fig-r-snr} shows the  fairness rate of the  $K$-user $2\times 2$ MIMO IC versus the SNR for $a_T=0.6$. 
As can be observed, the benefit of IGS is minor when $K=2$. 
However, by increasing the number of users, the performance improvement of IGS increases. The reason is that, by increasing the number of users, the interference level increases as well, which results in 
 more performance improvements by transmitting improper signals. 
Moreover, IGS performs much better at high SNR for $K=6$, similar to the   2-user SISO IC as depicted in Fig.  \ref{Fig-r-simo}a.

\begin{figure}[t!]
    \centering
    \begin{subfigure}[t]{0.5\textwidth}
        \centering
       \includegraphics{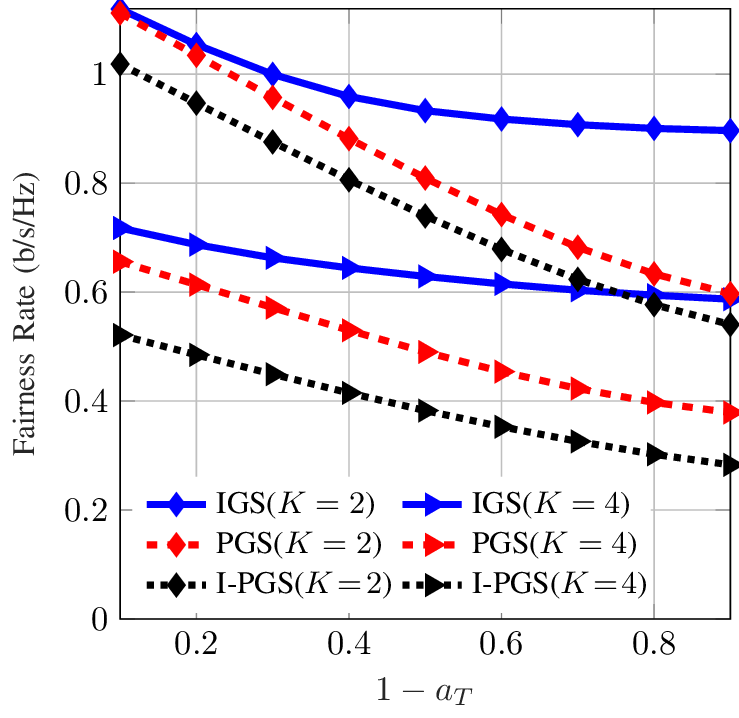}
        \caption{$K=2$ and $K=4$.}
    \end{subfigure}%
    \\ 
    \begin{subfigure}[t]{0.5\textwidth}
        \centering
           \includegraphics{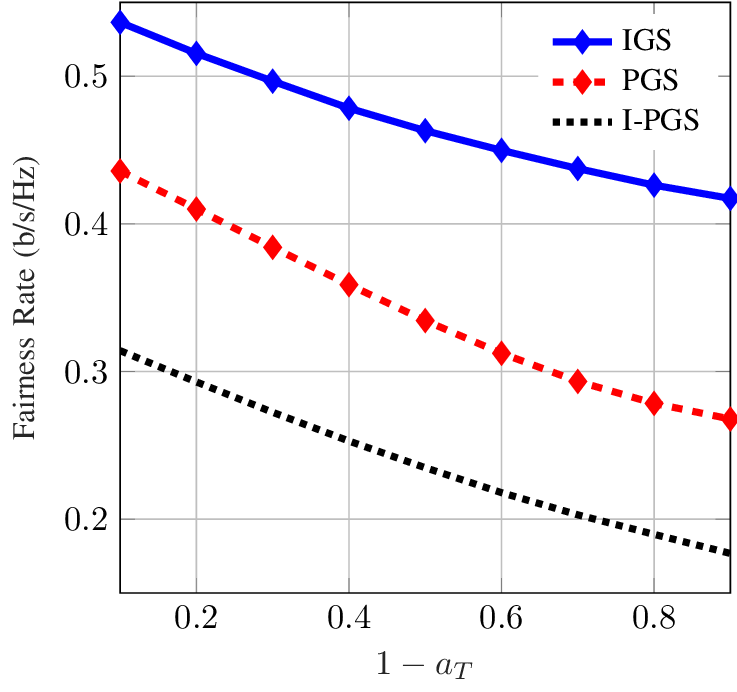}
        \caption{$K=6$.}
    \end{subfigure}
    \caption{Average fairness rate versus the I/Q imbalance level for the $K$-user $2\times 2$ MIMO IC with SNR$=0\,$dB.}
	\label{Fig-r-iq}
\end{figure}
Figure \ref{Fig-r-iq} shows the fairness rate versus the level of the I/Q imbalance,  $1-a_T$, for SNR$=0\,$dB and $N_T=N_R=2$. 
As can be observed, the IGS design is less affected by the HWI level for different $K$. 
When $K=2$, the IGS and PGS schemes perform very similarly in low HWI level. 
However, the performance of the PGS scheme drastically decreases with the HWI level, while the fairness rate of the IGS scheme decreases only slightly. 
When  $K=4$ and $K=6$, the same trend is observed, but the relative performance of the IGS scheme over the PGS scheme increases with $K$. 
Moreover, for a given $K$, the benefits of IGS increase with the level of the I/Q imbalance, as expected.

\begin{figure}[t!]
    \centering
    \begin{subfigure}[t]{0.5\textwidth}
        \centering
       \includegraphics{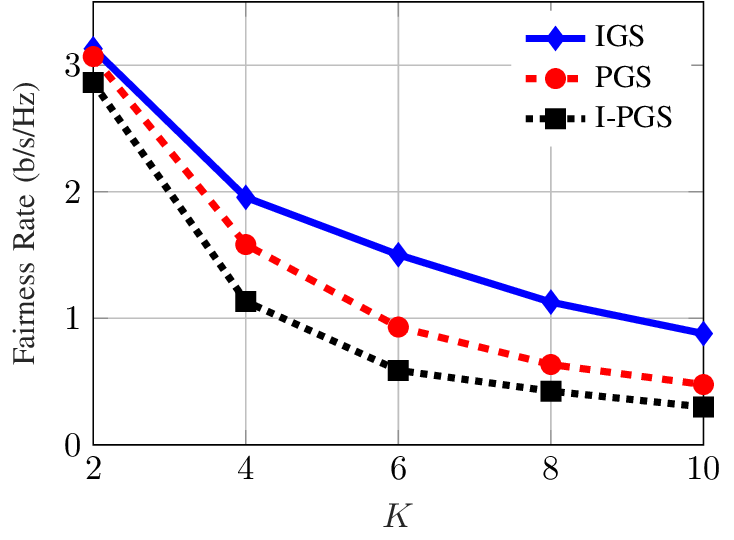}
        \caption{Achievable rate}
    \end{subfigure}%
    \\
    \begin{subfigure}[t]{0.5\textwidth}
        \centering
           \includegraphics{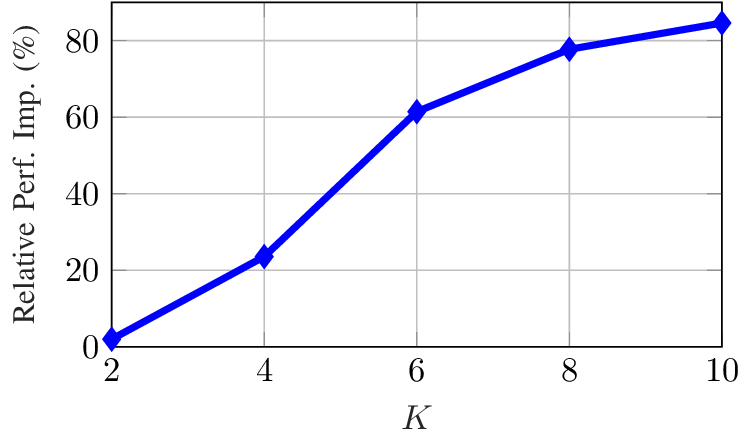}
        \caption{Relative improvement of IGS over PGS}
    \end{subfigure}
    \caption{Average fairness rate versus the number of users for the $K$-user $2\times 2$ MIMO IC with SNR$=10\,$dB.}
	\label{intersec}
\end{figure}
Figure \ref{intersec} considers the effect of the number of users on the fairness rate as well as on the performance of IGS for SNR$=10\,$dB, $a_T=0.6$ and $N_T=N_R=2$. As can be observed,
the fairness rates decreases when $K$ increases. 
Additionally, the relative performance improvement by IGS increases with the number of users.  For $K=10$, the relative improvement of IGS over PGS is more than $80\%$. 
The reason is that more users provoke more interference, which results in turn in more  improvements by IGS, as indicated before.  
\subsection{Achievable sum-rate}
\begin{figure}[t!]
    \centering
    \begin{subfigure}[t]{0.5\textwidth}
        \centering
       \includegraphics{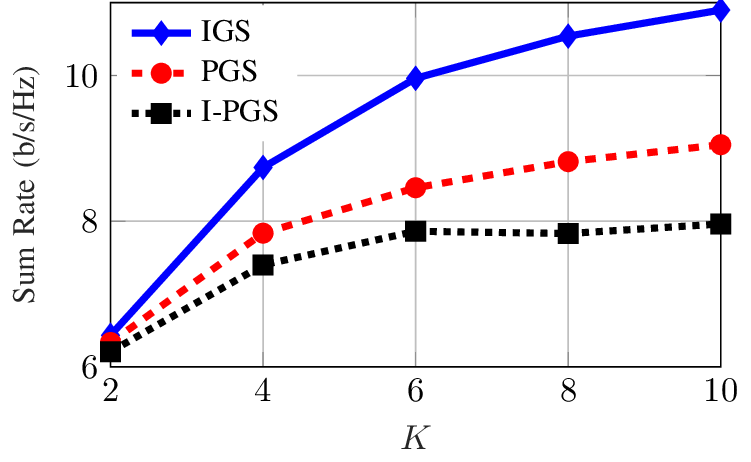}
        \caption{Achievable sum-rate}
    \end{subfigure}%
    \\
    \begin{subfigure}[t]{0.5\textwidth}
        \centering
           \includegraphics{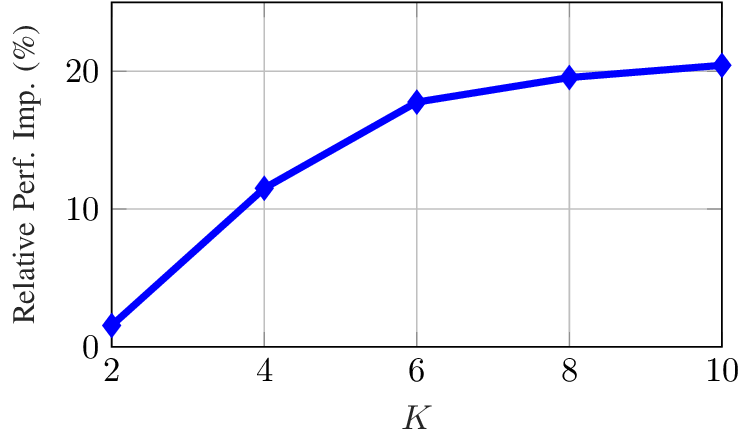}
        \caption{Relative improvement of IGS over PGS}
    \end{subfigure}
    \caption{Average achievable sum-rate and relative performance of IGS versus the number of users for the $K$-user $2\times 2$ MIMO IC with SNR$=10\,$dB.}
	\label{Fig1}
\end{figure}

In Fig. \ref{Fig1}, we show the effect of the number of users on the achievable sum-rate of the  $K$-user $2\times 2$ MIMO IC for SNR$=10\,$dB, and $a_T=0.6$. 
In this figure, we set the threshold in \eqref{rate-sum} to $R_{\text{th},k}=0$. 
We can observe that the sum-rate and also the relative performance improvement  of IGS over PGS are increasing with respect to $K$. 
Since we maximize the sum-rate without considering a QoS constraint, the rate of some users with weak direct links might be even 0, which causes less interference. 
As a result, the relative performance improvement of IGS is less significant than the improvements for the fairness rate observed in Fig.  \ref{intersec}. 

\subsection{Energy efficiency region}
In this subsection, we consider the EE region in \eqref{ee-region}. 
In general, IGS provides less EE benefits than rate benefits. 
For example, in \cite{soleymani2019energy}, it was shown that the conditions for the optimality of IGS over PGS in an UCR network are more stringent for EE than for rate. 
In other words, it might happen that IGS must be used for the SU rate to be maximized, while PGS must be used for the EE to be maximized.

\begin{figure}[t]
\centering
\includegraphics{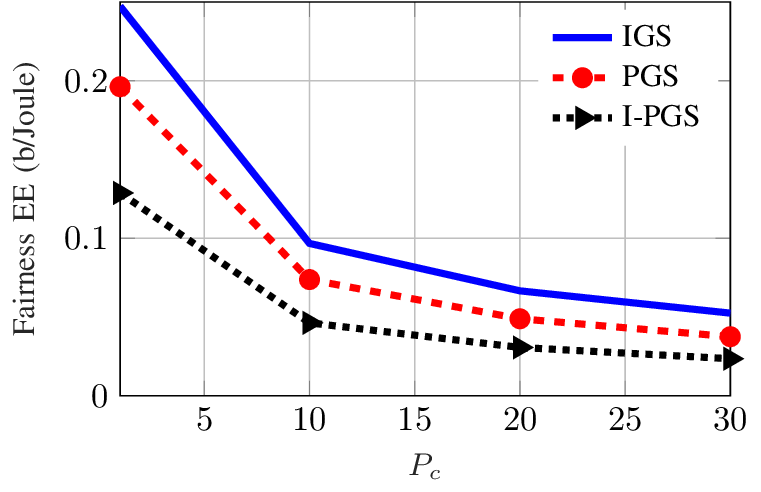}
\caption{Average fairness EE of the 6-user $2\times 2$ MIMO  IC versus  $P_c$.}
\label{ee-miso}
\end{figure}
\begin{figure}[t]
\centering
\includegraphics{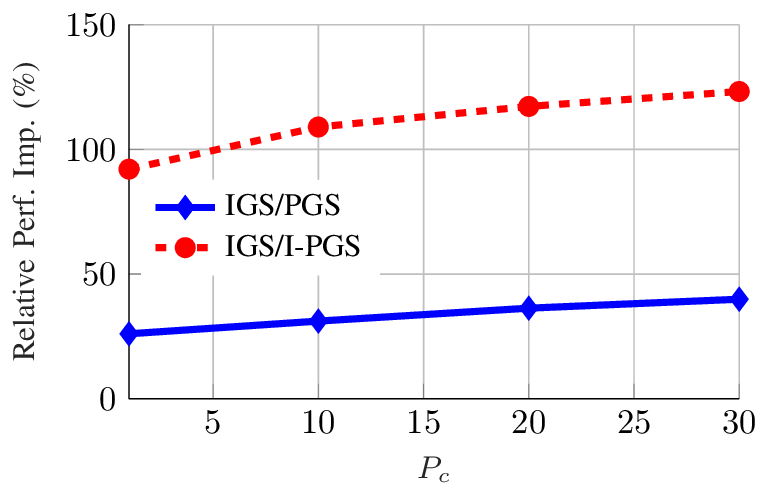}
\caption{Relative performance of IGS with respect to PGS and I-PGS versus $P_c$ for the 6-user $2\times 2$ MIMO  IC.}
\label{ee-b}
\end{figure}
In Fig. \ref{ee-miso}, we show the fairness EE of the  6-user $2\times 2$ MIMO IC versus $P_c$ for SNR$=10\,$dB,  and $a_T=0.6$. 
As can be observed, the fairness EE decreases with $P_c$. Moreover, the proposed IGS scheme outperforms the PGS scheme as well as I-PGS. 
Figure \ref{ee-b} shows the relative performance improvement by the IGS scheme with respect to the PGS and I-PGS schemes for the results in Fig. \ref{ee-miso}.  
As can be observed in these figures, the  fairness EE decreases with $P_c$; however, the benefits of employing IGS is increasing in $P_c$. 
The reason is that when $P_c$ is very large, the EE-region-optimization problem becomes equivalent to the achievable-rate-region problem, and as indicated, IGS can provide more gain in achievable-rate optimizations.

\subsection{Global Energy efficiency}
Figure \ref{gee} shows the global EE of the 6-user $2\times 2$ MIMO   IC versus $P_c$ for SNR$=10\,$dB,  and $a_T=0.6$. 
In this figure, we assume $R_{\text{th},k}=0$ in \eqref{ee-global}.
As can be observed, IGS provides minor benefits in terms of the global EE. 
Since the QoS constraint is not considered, it might happen that some  users are switched off, thus reducing the total level of interference. Moreover, the lower the interference level, the less need for advanced interference-management techniques. 
Thus, we can expect that the benefits of employing IGS decrease in global EE with respect to per-user EE. 
Note that IGS still performs slightly better than PGS, which can be due to the I/Q imbalance.
\begin{figure}[t]
\centering
\includegraphics{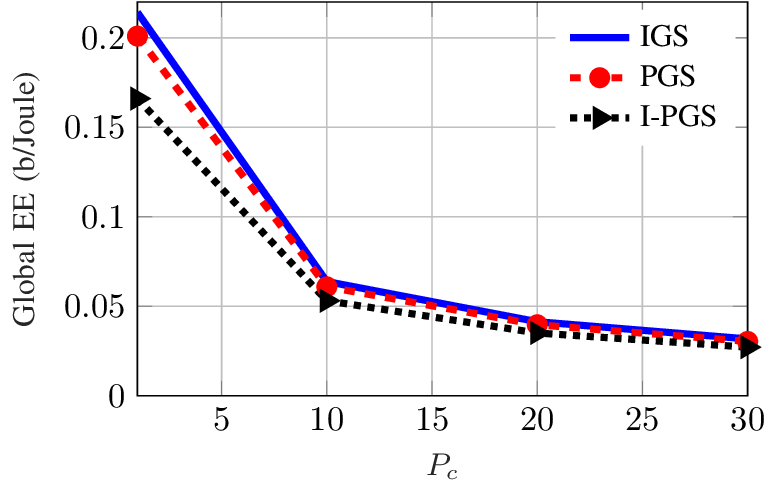}
\caption{Average global EE of the 6-user $2\times 2$ MIMO  IC versus $P_c$.}
\label{gee}
\end{figure}

\section{Conclusion}
This paper studied the performance  of IGS of a  $K$-user MIMO IC with HWI including I/Q imbalance at the transceivers. 
In the presence of I/Q imbalance, the received signal is a function of the widely linear transform of the transmitted signal and the aggregated noise. 
Hence, the effective noise is modeled as improper at the receiver side, which motivated us to consider the use of IGS.  
Considering achievable rates and EE as performance metrics, we proposed HWI-aware IGS schemes for the  $K$-user MIMO IC. 
We employed an optimization framework, which can obtain a stationary point of any optimization problem for interference-limited systems with TIN in which the objective function and/or constraints are linear functions of the achievable rate. 
In this paper, we derived a stationary point of the achievable rate-region, sum-rate maximization, EE region and global EE maximization problems. 
Our numerical results showed that IGS can improve the performance of the $K$-user MIMO IC with HWI from both achievable rate and EE points of view. We observed that the benefits of IGS as an interference-management technique increase with the number of users and decrease with the number of antennas. 
This is due to the fact that higher interference levels result in an increased need for interference management and consequently, more improvements by IGS.
We also observed that the benefit of employing IGS increases with impairment level.

As future work, it may be interesting to find out how close  the solution of this algorithm is to the corresponding global optimum solutions.
Additionally, our scheme is a centralized approach, which might not be applicable in some practical scenarios. Hence, distributed algorithms should also be developed. 

\section*{Acknowledgments}
The work of M. Soleymani and P. J. Schreier 
was supported by the German Research Foundation (DFG) under grant SCHR 1384/8-1. 
The work of I. Santamaria was supported by Ministerio de Ciencia e Innovacion of Spain, and AEI/FEDER funds of the E.U., under grants TEC2016-75067-C4-4-R (CARMEN) and  PID2019-104958RB-C43 (ADELE).
The authors would like to thank Dr. C. Lameiro for discussions on the numerical results.

\bibliographystyle{IEEEtran}
\bibliography{ref2}

\end{document}